\documentclass[letterpaper, 10 pt, conference]{ieeeconf}

\usepackage{amsmath}
\usepackage{amssymb}
\usepackage{amsfonts}
\usepackage{graphicx}
\usepackage{float}
\usepackage{graphics}
\usepackage{graphicx}
\usepackage{theorem}
\usepackage{euscript}
\usepackage{psfrag}

\newcommand{\abs}[1]{\left\vert#1\right\vert}

\newcommand{\vct}[1]{\boldsymbol{#1}}
\newcommand{\Mat}[1]{\boldsymbol{#1}}

\DeclareMathOperator{\diag}{diag}

\def\E{{\mathbb E}}        

\newcommand{\Sinf}{S_\infty}
\newcommand{\Linf}{\mathcal{L}_\infty}

\newcommand{\myscale}{0.28}

\DeclareMathOperator{\tr}{tr}

\theoremstyle{plain}
\newtheorem{thm}{Theorem}
\newtheorem{cor}[thm]{Corollary}

\newtheorem{prop}[thm]{Proposition}

\begin{document}

\title{The Two-Tap Input-Erasure Gaussian Channel and its Application to Cellular Communications}

\author{\authorblockN{O. Somekh\authorrefmark{1}, O. Simeone\authorrefmark{2}, H. V. Poor\authorrefmark{1}, and S. Shamai (Shitz)\authorrefmark{3}}\\
\authorblockA{\authorrefmark{1}
Department of Electrical Engineering, Princeton University, Princeton, NJ 08544, USA, {\{orens, poor\}@princeton.edu}}
\authorblockA{\authorrefmark{2}
CWCSPR, New Jersey Institute of Technology, Newark, NJ 07102-1982, USA, {osvaldo.simeone@njit.edu.}}
\authorblockA{\authorrefmark{3}
Department of Electrical Engineering,Technion, Haifa, 32000, Israel, {sshlomo@ee.technion.ac.il}}
}

\maketitle

\begin{abstract}
This paper considers the input-erasure Gaussian channel. In contrast to the output-erasure channel where erasures are applied to the output of a linear time-invariant (LTI) system, here erasures, known to the receiver, are applied to the inputs of the LTI system. Focusing on the case where the input symbols are independent and identically distributed (i.i.d)., it is shown that the two channels (input- and output-erasure) are equivalent. Furthermore, assuming that the LTI system consists of a two-tap finite impulse response (FIR) filter, and using simple properties of tri-diagonal matrices, an achievable rate expression is presented in the form of an infinite sum. The results are then used to study the benefits of joint multicell processing (MCP) over single-cell processing (SCP) in a simple linear cellular uplink, where each mobile terminal is received by only the two nearby base-stations (BSs). Specifically, the analysis accounts for ergodic shadowing that simultaneously blocks the mobile terminal (MT) signal from being received by the two BS. It is shown that the resulting ergodic per-cell capacity with optimal MCP is equivalent to that of the two-tap input-erasure channel. Finally, the same cellular uplink is addressed by accounting for dynamic user activity, which is modelled by assuming that each MT is randomly selected to be active or to remain silent throughout the whole transmission block. For this alternative model, a similar equivalence results to the input-erasure channel are reported. 

\end{abstract}

\section{Introduction}

In recent years various variants of erasure networks have been the focal point of many works in the fields of information theory, coding theory, communication, and digital signal processing (DSP). It turns out that the idealization of situations where symbols are detectable with very high or very low probability is useful in providing accurate and tractable assessment of practical interest. Applications of erasure networks, in general, and erasure channels, in particular, include packet-switched networks, wireless fading channels, impulse noise channels, and sensor networks with faulty transducers. Another interesting DSP application is decimation using frequency shaped random erasures \cite{Dey-Oppenheim-ICASSP07} (see also \cite{Dey-Russell-Oppenheim-SP06} for a somewhat dual problem).

\begin{figure}[tb]
\begin{center}
\includegraphics[scale = 0.3, angle=90]{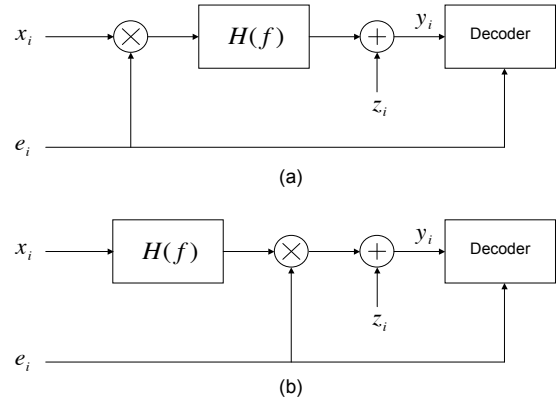}
\vspace{-0.5cm}
\caption{The linear input- and output-erasure Gaussian channels.} \label{fig: IEGC OEGC schemes}
\end{center}
\end{figure}

The simplest non-trivial setting incorporating erasures is perhaps a memoryless discrete channel observed through a memoryless erasure channel. In this case, it is straightforward to see that the capacity of the concatenated channel is the capacity of the discrete channel times the erasure rate \cite{Julian-ISIT02}. Moreover, even if the erasures have memory, the latter result holds \cite{Julian-ISIT02} \cite{Verdu-Weissman-IT08}. However, adding memory to the discrete channel changes the picture considerably.

In \cite{Tulino-Verdu-Caire-Shamai_ISIT07}, Tulino et al. have calculated the capacity of an input power-constrained Gaussian channel concatenated with a memoryless erasure channel. The problem requires obtaining the asymptotic spectral distribution of a sub-matrix of a nonnegative definite Toeplitz matrix obtained by randomly deleting rows and columns. In addition, \cite{Tulino-Verdu-Caire-Shamai_ISIT07} presents an explicit expression for the optimal input power spectral density (PSD) function. Asymptotic capacity expressions in various regimes such as extreme-SNR, sporadic erasures events, and sporadic non-erasure events, are reported as well. In addition, several useful bounds (lower and upper) in closed form are presented. Recently, the results of \cite{Tulino-Verdu-Caire-Shamai_ISIT07} were extended in \cite{Tulino-Verdu-Caire-Shamai-ISIT08} to include any flat fading distribution, not necessarily the on-off process obtained by erasures.

In this paper, we consider the input-erasure Gaussian channel. According to this setup, erasures (known at the receiver only) are applied directly to the Gaussian input of a linear time-invariante (LTI) system (e.g. a finite input response filter, FIR), whose output samples are affected by additive Gaussian noise. To simplify the analysis, we consider independent and identically distributed (i.i.d.) input signaling, and show that in this case the input-erasure channel and the output-erasure channel of \cite{Tulino-Verdu-Caire-Shamai_ISIT07} are equivalent. Focusing on the case where the channel filter includes two taps, we present an alternative analysis to \cite{Tulino-Verdu-Caire-Shamai_ISIT07} that obtains the achievable rate without relying on random matrix theory tools. The analysis is based on a simple observation that in the two-tap case, each erasure splits the output vector into two independent sub-vectors. In addition, it uses a useful property of tri-diagonal matrices, which states that their determinant can be expressed as a weighted sum of the determinants of their two main principal sub-matrices. It is also shown that the analysis can be extended to the case in which the erasures are not memoryless.

The results are then used to find the achievable rate supported by a simple cellular uplink model (where inter-cell interference is limited only to adjacent cells), in which the transmission of each mobile terminal (MT) undergoes erasures due to\textit{ ergodic shadowing}, which simultaneously blocks it from being received by all the base-stations (BSs). We show that with optimal joint decoding of the received signals from all the BSs (that is, multicell processing, MCP), the ergodic per-cell capacity of this cellular uplink equals the rate of the two-tap input-erasure channel considered earlier. The performance of MCP is compared to that of single-cell processing (SCP), demonstrating that the benefits of MCP decrease with the erasure rate (the reader is referred to \cite{Somekh-Simeone-Barness-Haimovich-Shamai-BookChapt-07} and \cite{Shamai-Somekh-Zaidel-JWCC-2004} for recent surveys on MCP). We then consider a similar cellular uplink setup but with a \textit{dynamic user activity} model, in which each user is randomly active or silent for the whole transmission block. We show that with MCP the system throughput per active user is equal to the rate of the two-tap input-erasure channel divided by the probability of an arbitrary user being active. Similar benefits of MCP over SCP, as those reported for the shadowing scenario, are observed also for the user activity model. 

It is finally noted that a related analysis has been recently reported in \cite{Levy-Zeitouni-Shamai-2Tap_UP08}, where the \textit{outage} capacity supported by a similar cellular uplink and user activity model was derived with both MCP and SCP.

\section{Channel Model}

In this work we consider a channel (see Fig. \ref{fig: IEGC OEGC schemes}.a), where erasures $\{e_{n}\}$, \emph{known to the receiver only}, are applied to the inputs symbols $\{x_{n}\}$ of a discrete time LTI system with transfer function $H(f)$. For simplicity, we assume that the LTI system is causal and has a finite\footnote{It is easily verified that the results extend to any LTI system with finite gain $\int_{0}^{1}\left\vert H(f)\right\vert ^{2}df<\infty$.} impulse response of length $(L+1)$ symbols. Accordingly, the received signal at time index $n$, is given by
\begin{equation}
\label{ }
y_n=\sum_{m=0}^{L}h_{m}x_{n-m}e_{n-m}+z_n\ ,
\end{equation}
where $\{h_m\}_{m=0}^L$ are the impulse response coefficients such that
\begin{equation}
\label{ }
H(f)=\sum_{m=0}^L e^{-j2\pi m f}h_m\ ,
\end{equation}
and $z_n$ is the additive complex Gaussian noise $z_n\sim\mathcal{CN}(0,1)$. The erasures 
$e_n$, are assumed to be i.i.d. Bernoulli distributed random variables (r.v.'s)  $e_n\sim\mathcal{B}(q)$ (i.e., $e_n\in\{0,1\},\ Pr(e_n=0)=q,\ Pr(e_n=1)=1-q$). 

We consider transmissions in blocks of $N$ symbols. The received vector of an arbitrary block is given by 
\begin{equation}
\label{eq: IEGC output vector}
Y_{N+L} =  \Mat{H}_{N+L}\Mat{E}_N X_{N}+Z_{N+L}\ ,
\end{equation}
where $Y_{N+L} \triangleq(y_1,\ldots,y_{N+L})^\dagger$ is the $(N+L)\times 1$ output vector, $X_N\triangleq(x_1,\ldots,x_N)^\dagger$ is the complex Gaussian $N\times 1$ input vector $X_N\sim\mathcal{CN}(\Mat{0},\Mat{Q}_N)$ with average power constraint $\tr{(\Mat{Q}_N)}\le NP$,  $\Mat{E}_N\triangleq\diag(e_1,\ldots,e_N)$ is the $N\times N$ diagonal erasure matrix, and $Z_{N+L} \triangleq(z_1,\ldots,z_{N+1})^\dagger$ is the $(N+L)\times 1$ noise vector $Z_{N+L}\sim\mathcal{CN}(\Mat{0},\Mat{I}_{N+L})$. In addition, $\Mat{H}_{N+L}$ is the $(N+L)\times N$ $L$-diagonal \emph{Toeplitz} channel transfer matrix with $[\Mat{H}_{N+L}]_{n,m}=h_{n-m}$ (where out-of-range indices should be ignored). 

Recalling that the receiver is aware of the erasures, it is easy to prove that a Gaussian input distribution is a capacity-achieving distribution and that the capacity is given by
\begin{multline}
\label{eq:  IEGC capacity implicit}
C=\lim_{N\rightarrow\infty}\max_{\Mat{Q}_N,\ \tr{(\Mat{Q}_N)}\le NP} \\ \E\left[\frac{1}{N}\log \det{\left(\Mat{I}_{N+L}+
\Mat{H}_{N+L}\Mat{E}_N\Mat{Q_{N}}\Mat{E}^\dagger_N\Mat{H}^\dagger_{N+L}\right)}\right]\ ,
\end{multline} 
where the expectation is taken over the diagonal entries of the erasure matrix $\Mat{E}_N$, and $\Mat{I}_{N+L}$ is an $(N+L)\times (N+L)$ identity matrix.

Unfortunately, finding the optimal covariance matrix and calculating \eqref{eq:  IEGC capacity implicit} is still an open problem. Instead, we focus in the rest of this work on assessing achievable rates resulting by using i.i.d. inputs, hence, by setting $\Mat{Q}_N=P\Mat{I}_N$. 

\section{Rate Analysis}\label{sec: IEGC rate analysis} 

\subsection{Arbitrary FIR Filter}

Here we consider an arbitrary $L$-length FIR filter and rewrite the logdet term of \eqref{eq:  IEGC capacity implicit} as:
\begin{multline}
\label{eq: IEGC-OEGC equivalence}
\log \det{\left(\Mat{I}_{N+L}+
P\Mat{H}_{N+L}\Mat{E}_N\Mat{E}^\dagger_N\Mat{H}^\dagger_{N+L}\right)}=\\
\log \det{\left(\Mat{I}_{N}+
P\Mat{E}^\dagger_N\Mat{H}^\dagger_{N+L}\Mat{H}_{N+L}\Mat{E}_N\right)}\ .
\end{multline}
By noting that the left-hand-side (LHS) of the last equation is the same as the one resulting from the analysis of the output-erasure channel with i.i.d. inputs, considered in \cite{Tulino-Verdu-Caire-Shamai_ISIT07} (see Fig. \ref{fig: IEGC OEGC schemes}.b), we can state the following.
\begin{prop}
The Gaussian input-erasure and output-erasure channels with i.i.d. inputs are equivalent.
\end{prop}
This result implies that with i.i.d. inputs, both channels achieve the same rate. Hence, all the results reported in \cite{Tulino-Verdu-Caire-Shamai_ISIT07} hold verbatim for the input-erasure channel with i.i.d. inputs, by setting the PSD accordingly (i.e., $S(f)=P\abs{H(f)}^2$) in the results of  \cite{Tulino-Verdu-Caire-Shamai_ISIT07}.


The analysis in \cite{Tulino-Verdu-Caire-Shamai_ISIT07} uses random matrix theory tools to derive the main result which is the rate of the output-erasure channel for any channel FIR filter and any input PSD. The rate result of  \cite{Tulino-Verdu-Caire-Shamai_ISIT07} involves a fixed-point integral equation and in general it is not explicitly formulated. Next, we present an alternative analysis for the special case of i.i.d. inputs and a two-tap channel filter. The resulting rate expression is considerably different than that of  \cite{Tulino-Verdu-Caire-Shamai_ISIT07}, and may shed some further light on the problem.

\subsection{Two-Tap Filter ($L=1$)}

In this section we consider a two-tap channel filter (i.e., its impulse-response includes only two non-zero coefficient $h_0$, and $h_1$). Hence, the channel transfer matrix $\Mat{H}_{N+1}$ reduces to a bi-diagonal Toeplitz matrix. For this special case we have the following.
\begin{prop}\label{prop: two-tap IEGC rate}
The achievable rate of the two-tap input-erasure Gaussian channel with i.i.d. inputs is given by
\begin{equation}
\label{eq: two tap IEGC rate}
R_{\mathrm{2tap}} = q^2\sum_{n=1}^\infty (1-q)^n \log\left(\frac{r^{n+1}-s^{n+1}}{r-s}\right)\ ,
\end{equation}
where
\begin{equation}
\label{eq: a b definitions}
a \triangleq 1+P (\abs{h_0}^2+\abs{h_1}^2)\quad\mathrm{and} \quad b \triangleq P\abs{h_0}\abs{h_1}\ ,
\end{equation}
and
\begin{equation}
\label{eq: r s definitions}
r\triangleq\frac{1}{2}\left(a+\sqrt{a^2-4b^2}\right)\ \ \mathrm{and}\ \ s\triangleq\frac{1}{2}\left(a-\sqrt{a^2-4b^2}\right)\ .
\end{equation}
\end{prop}
\begin{proof}
See Appendix 
\end{proof}
The proof relies on a simple observation that in the two-tap case, each erasure splits the output vector into two independent sub-vectors. In addition, it uses a useful property of tri-diagonal matrices, which states that their determinant can be expressed as a weighted sum of their first two main principal sub-matrices. It is noted that this analysis does not involve classical random matrix theory. Moreover, since it relies heavily on the two-tap assumption, it cannot be applied to input-erasure channels with filter lengths larger than two.

Examining the rate expression \eqref{eq: two tap IEGC rate}, a few comments are in place. Although the rate expression is an infinite sum, any finite sum, which is evidently a lower bound, can be calculated numerically in a straightforward manner. In addition, convergence is assured since the expression is a power series and $(1-q)\leq 1$. It is also observed that the rate is independent of the filter coefficients' phases, and hence, random phases (known to the receiver) can be added to the channel model without changing the end result. By applying Hadamard's inequality directly to the covariance matrix $(\Mat{I}_N+\Mat{G}_N)$ (see \eqref{eq: Cov matrix def}), it can be shown that for a fixed filter gain, filter memory reduces the rate. Hence, the rate is maximized for a one-tap filter with equal gain. On this note, 
it is easily verified that for the special case of $h_{1}=0$ (memoryless Gaussian channel with erasures), $s=0$ and \eqref{eq: two tap IEGC rate} reduces to
\begin{equation}
\label{ }
R_{\mathrm{1tap}} = (1-q) \log(1+P\abs{h_0}^2)\ ,
\end{equation}
which was already reported in \cite{Julian-ISIT03}. Another special case, which is also an upper bound, is the case where $q=0$. The capacity of this two-tap channel with i.i.d. inputs is given by \cite{Somekh-Zaidel-Shamai-IT-2007-CWIT-2005}
\begin{equation}
\label{eq: 2tap ub}
R_{\mathrm{2tap-ub}}=\log r\  ,
\end{equation}
where $r$ is defined in \eqref{eq: r s definitions}.

Turning to asymptotic analyses, we refer the reader to \cite{Tulino-Verdu-Caire-Shamai_ISIT07}, where compact closed-form expressions are reported for the rate low-SNR characterization, the rate with sporadic erasures, and the rate with sporadic non-erasures. These expressions can be used verbatim for the two-tap input-erasure channel of interest by setting i.i.d. inputs in \cite{Tulino-Verdu-Caire-Shamai_ISIT07}. In addition, \cite{Tulino-Verdu-Caire-Shamai_ISIT07} provides also the high-SNR characterization of the output-erasure channel rate, albeit the calculation of the high-SNR power offset involves a fixed point-equation. 
Applying definitions of \cite{Lozano-Tulino-Verdu-high-SNR-IT05} to the rate expression \eqref{eq: two tap IEGC rate}, we provide an alternative expression for the high-SNR power offset as well.
\begin{prop}
The two-tap input-erasure Gaussian channel rate with i.i.d. inputs is characterized in the high-SNR regime by
\begin{equation}
\begin{aligned}
\Sinf&=(1-q)\quad\mathrm{and}\\
\Linf &= -q^2\sum_{n=1}^\infty (1-q)^{n-1}\log\left(\frac{ \abs{h_0}^{2(n+1)}-\abs{h_1}^{2(n+1)}} { \abs{h_0}^2- \abs{h_1}^2}\right).
\end{aligned}
\end{equation}
\end{prop}
And the deleterious effects of increasing erasure rate $q$ are clearly visible.

\subsubsection{Erasures with Memory}
Following the proof of Prop. \ref{prop: two-tap IEGC rate}, it can be verified that the presented analysis method can handle also erasures with memory (stationary and ergodic). The only stage that should be altered is the calculation of  \eqref{eq: consec ones prob}, that is, the probability of having an isolated sequence $\ell$ of $n$ consecutive non-erasures should be recalculated for the new erasure source. For example, assuming the erasures follow a first order Markov chain with transition probabilities $Pr(0\rightarrow 1)=1-q_0$, and $Pr(1\rightarrow 0)=q_1$, it is easily verified that
\begin{equation}
\label{ }
Pr(\ell = n)=\frac{q_1^2(1-q_0)}{1-q_0+q_1}(1-q_1)^{n-1}\ ,
\end{equation}
and the corresponding rate of the two-tap input erasure channel with i.i.d. inputs is given by
\begin{equation}
\label{eq: two tap IEGC markov erasure rate}
R_{\mathrm{2tap}} = \frac{q_1^2(1-q_0)}{1-q_0+q_1}\sum_{n=1}^\infty (1-q_1)^{n-1} \log\left(\frac{r^{n+1}-s^{n+1}}{r-s}\right)\ .
\end{equation}
As expected, for $q_0=q_1=q$ the erasure source is i.i.d. and \eqref{eq: two tap IEGC markov erasure rate} reduces to \eqref{eq: two tap IEGC rate}.
It is interesting to note that for the special case $h_1=0$, which corresponds to a memoryless channel followed by an erasure channel with memory, the rate reduces to 
\begin{equation}
\label{ }
R_{\mathrm{1tap}} = (1-q) \log(1+P\abs{h_0}^2)\  ,
\end{equation}
where
\begin{equation}
\label{ }
q = \frac{q_1}{1-q_0+q_1}\ ,
\end{equation}
is the steady-state erasure rate. This results is in agreement with \cite{Verdu-Weissman-IT08}, where it has been shown that the capacity of a concatenation of a discrete memoryless channel with capacity $C$, and an erasure channel (possibly with memory) with erasure rate $q$, is equal to $(1-q)C$. 

\section{Numerical Results}\label{sec: numerical results}

\begin{figure}
\begin{center}
\includegraphics[scale = \myscale]{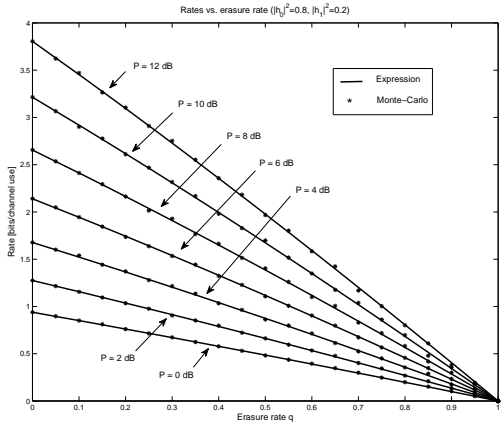}
\vspace{0cm}
\caption{Rates vs. erasure rate $q$, for $P=0,\ 2,\ 4,\ 6,\ 8,\ 10,\ 12$, and $\abs{h_0}^2=0.8,\ \abs{h_1}^2=0.2$.} \label{fig: IEGC rate vs q 1}
\end{center}
\end{figure}

\begin{figure}
\begin{center}
\includegraphics[scale = \myscale]{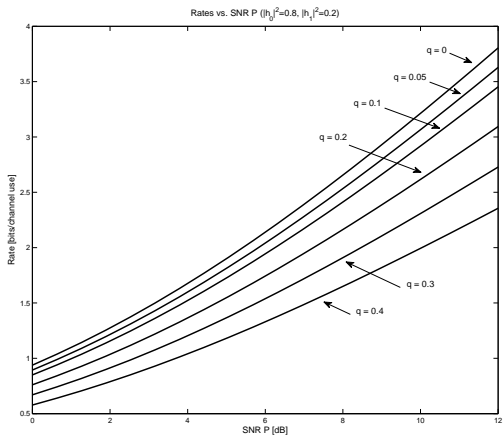}
\vspace{0cm}
\caption{Rates vs. SNR $P$, for $q=0,\ 0.05,\ 0.1,\ 0.2,\ 0.3,\ 0.4$, and $\abs{h_0}^2=0.8,\ \abs{h_1}^2=0.2$.} \label{fig: IEGC rate vs snr 1}
\end{center}
\end{figure}

\begin{figure}
\begin{center}
\includegraphics[scale = \myscale]{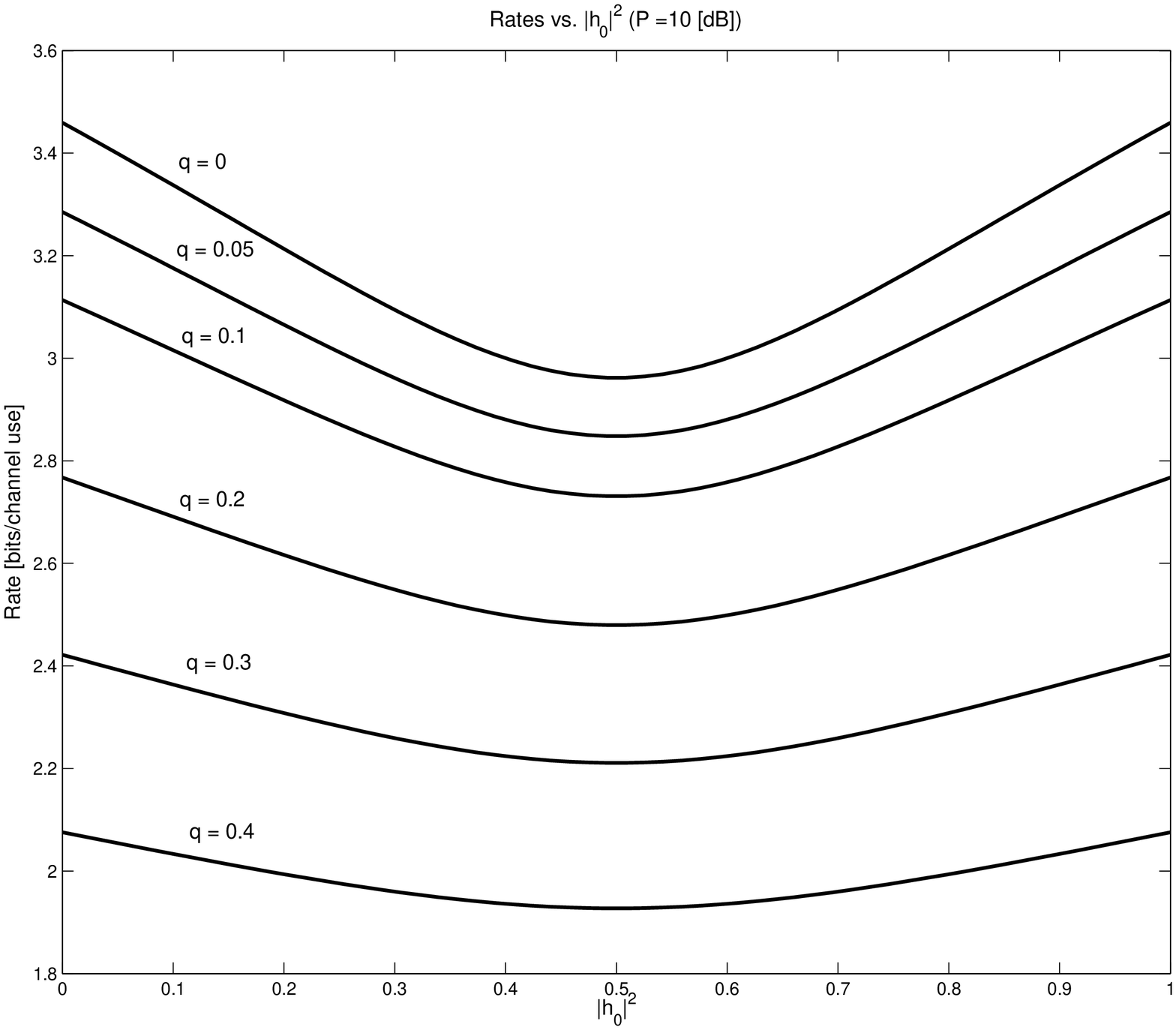}
\vspace{0cm}
\caption{Rates vs. filter coefficient $\abs{h_0}^2$, for $q=0,\ 0.05,\ 0.1,\ 0.2,\ 0.3,\ 0.4$.} \label{fig: IEGC rate vs h0 1}
\end{center}
\end{figure}

The rate (\ref{eq: two tap IEGC rate}) is plotted in Figure \ref{fig: IEGC rate vs q 1} as a function of the erasure rate $q$, for several SNR values, and filter coefficients $\left\vert h_{0}\right\vert
^{2}=0.8,\ \left\vert h_{1}\right\vert ^{2}=0.2$. Expression \eqref{eq: two tap IEGC rate} is calculated by taking 200 summands and is plotted with solid lines, while Monte-Carlo simulation results of 50 trials, each with block size of 200 symbols, are shown as well. Examining the figure, the deleterious effect of erasures is evident. In addition, the figure demonstrates a good match between the Monte-Carlo simulations and the \textquotedblleft exact\textquotedblright\ results. Similar conclusions can be obtained from Figure \ref{fig: IEGC rate vs snr 1}, which includes rate curves (calculated from (\ref{eq: two tap IEGC rate}) as above) as functions of the SNR $P$ for several values of the erasure rate $q$, and $\left\vert h_{0}\right\vert ^{2}=0.8,\ \left\vert h_{1}\right\vert ^{2}=0.2$. In addition, the erasure-free ($q=0$) rate curve \eqref{eq: 2tap ub} is included as a reference.

In Figure \ref{fig: IEGC rate vs h0 1} the rate curves are plotted as functions of the filter coefficients $\left\vert h_{0}\right\vert ^{2}$ for a unit gain filter (i.e., $\left\vert h_{0}\right\vert ^{2}+\left\vert h_{1}\right\vert ^{2}=1$), and $P=10$ [dB]. It is observed that the worst-case filtering corresponds to coefficients with equal amplitudes (i.e., $\left\vert h_{0}\right\vert ^{2}=\left\vert h_{1}\right\vert ^{2}$). On the other hand, as claimed earlier, best-case filtering is achieved for flat or memoryless channels (i.e., $\left\vert h_{1}\right\vert ^{2}=0$ or $\left\vert h_{1}\right\vert ^{2}=0$). Moreover, these observations are independent of the erasure rate $q$.

\section{Application to Cellular Communications}

Cellular systems are currently the main media to provide high-data rate services to mobile users. Therefore, a continued search for techniques that provide better service and coverage in cellular systems is under way. A promising cooperation-based technique is joint MCP, proposed by Wyner in \cite{Wyner-94}. Accordingly, clusters of BSs jointly process their signals to mitigate or eliminate the overall interference, since more signals are useful. Unlike conventional approaches that treats interference stemming from other cells as noise, or tries to avoids it, this approach (also referred to as \emph{distributed antenna array}) exploits the interference to provide better rates to the clusters' MTs.

Here we are interested in further exploring the benefits of MCP when ergodic erasures affect the MTs' transmissions in the uplink. Such erasures account for a scenario where an MT may undergo ergodic deep fading (or shadowing) such that its transmission cannot be received by \emph{any} of its cluster's antennas.

Motivated by the fact that inter-cell interference is essentially limited to only a small number of BSs, and with mathematical tractability in mind, we assume a synchronous Wyner-like cellular uplink setup. According to this setup (see Fig. \ref{fig: SHO model}), an infinite number of cells are arranged on a line, as along a highway. Moreover, the each MT ``sees'' its own BS antenna (with unit path gain), and the adjacent BS antenna (with path gain $\alpha\in\lbrack0,\ 1]$) only. This model, which was first introduced in \cite{Somekh-Zaidel-Shamai-IT-2007-CWIT-2005}, focuses on users lying on the cells' boundaries and is thus sometimes referred to as the \emph{soft-handoff} model. Specifically, assuming a single active user per-cell (intra-cell time-division multiple-access), the received signal at the $n$th BS antenna, for an arbitrary time index, is given by
\begin{equation}
\label{ }
y_n = e_n x_n+e_{n-1}\alpha x_{n-1}+z_n\ ,
\end{equation}
where $x_n$ and $x_{n-1}$ are the MTs' transmissions $x_n,\ x_{n-1}\sim \mathcal{CN}(0,P)$, $e_n$ and $e_{n-1}$ are the erasures $e_n,\ e_{n-1}\sim \mathcal{B}(q)$, and $z_n$ is the additive noise $z_n  \sim \mathcal{CN}(0,1)$. It is noted that the erasures are assumed to be i.i.d. among different users (along the cells' indices $n$), and ergodic along the time index for each user. Finally, it is noted that users are not allowed to cooperate.

\begin{figure}
\begin{center}
\includegraphics[angle = 90, scale = 0.305]{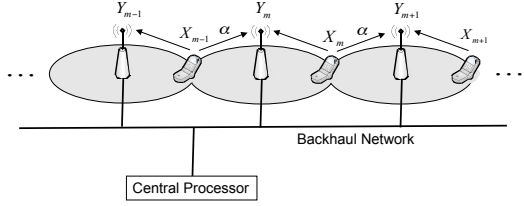}
\vspace{-1.5cm}
\caption{The soft-handoff cellular uplink model.} \label{fig: SHO model}
\end{center}
\end{figure}

\subsection{Single Cell Processing}

As a reference results we consider two conventional SCP approaches. 

\subsubsection{No Frequency Planning}

The simplest approach is for each cell to decode its own signals while treating the signals from the interfering cell as noise. Assuming that each cell is aware of its own user's erasures and those of the interfering cell's user (see \cite{Shamai-Wyner-97-I} for a similar setting), it is easily verified that the achievable rate is given by
\begin{equation}
\label{eq: spc rate raw}
R_{\mathrm{scp}} = \E\left[\log\left(1+\frac{e_{\mathrm{sig}}P}{1+\alpha^2e_{\mathrm{int}}P}\right)\right]\ ,
\end{equation}
where the expectation is taken over $e_{\mathrm{sig}}$ and $e_{\mathrm{int}}$ which represent arbitrary i.i.d. erasure r.v.'s of the useful signal and interference, respectively, $e_{\mathrm{sig}},\ e_{\mathrm{int}}\sim \mathcal{B}(q)$. By straightforward calculation we get that \eqref{eq: spc rate raw} is given by
\begin{equation}
\label{ }
R_{\mathrm{scp}}=(1-q)\log\left(\frac{(1+P)^q\left(1+(1+\alpha^2)P\right)^{1-q}}{(1+\alpha^2P)^{1-q}}\right)\ .
\end{equation}
The high-SNR characterization of $R_{\mathrm{scp}}$ is given by
\begin{equation}
\label{ }
\Sinf= q(1-q)\quad\mathrm{and}\quad \Linf = 0\ .
\end{equation}
It is interesting to observe that the rate above demonstrates a \textit{non-interference limited behavior}, although the receiver treats the other cell's signals as noise. This is easily explained by noting that with probability $q(1-q)$ the useful signal is received without interference.

\subsubsection{Inter-Cell Frequency Sharing}

Here we consider a more involved approach that requires static resource sharing between the cells. Since, according to the soft-handoff model, interference stems from one cell only, by dividing the available bandwidth of $W$ [Hz] into two orthogonal subbands with bandwidth of $W/2$ [Hz] each, and assigning them alternately to even and odd indexed cells, interference is totally avoided. Accordingly, it is easily verified that the achievable rate is given by\begin{equation}
\label{ }
R_{\mathrm{icfs}}=\frac{1}{2}\E[\log(1+eP)]=\frac{1-q}{2}\log(1+P)\ ,
\end{equation}
where $e$ is an arbitrary erasure r.v., $e\sim \mathcal{B}(q)$. 

The high-SNR characterization of $R_{\mathrm{icfs}}$ is given by
\begin{equation}
\label{ }
\Sinf= \frac{1-q}{2}\quad\mathrm{and}\quad \Linf = 0\ ,
\end{equation}
where the deleterious effects of both the erasures, and the bandwidth sharing are clearly observed in the multiplexing gain expression. Comparing the multiplexing gains of the SCP and inter-cell frequency sharing (ICFS) rates, reveals the superiority of the latter for erasure rates $q<1/2$ in the high-SNR regime.

\subsection{Multicell Processing}

With MCP, the BSs send their received signals to a central processor (CP) via an ideal backhaul network. The CP collects the received signals and jointly decodes the MTs' messages. Since the CP is aware of all the erasures, the overall channel is an ergodic multiple-access channel (MAC) for which the sum-rate capacity divided by the number of cells is given by
\begin{multline}
\label{eq: MCP rate}
C_{\mathrm{mcp}}=\lim_{N\rightarrow\infty}\\ \E\left[\frac{1}{N}\log \det{\left(\Mat{I}_{N+1}+
P\Mat{H}_{N+1}\Mat{E}_N\Mat{E}^\dagger_N\Mat{H}^\dagger_{N+1}\right)}\right]\ ,
\end{multline}
where $\Mat{E}_N\triangleq\diag(e_1,\ldots,e_N)$ is the $N\times N$ diagonal erasure matrix, and $\Mat{H}_{N+1}$ is an $(N+1)\times N$ bi-diagonal \emph{Toeplitz} channel transfer matrix with $[\Mat{H}_{N+1}]_{n,n}=1$ and $[\Mat{H}_{N+1}]_{n,n-1}=\alpha$ (where out-of-range indices should be ignored). It is noted that, by the symmetry of the model, the ergodic per-cell sum-capacity derived above is also the equal per-cell capacity. 

Comparing the rate expression of the two-tap input-erasure channel with i.i.d. inputs (see (\ref{eq:  IEGC capacity implicit})) to the per-cell sum-rate supported by the soft-handoff model, the following is evident.
\begin{cor}\label{corr: MCP corr}
The per-cell sum-rate supported by the soft-handoff model with optimal joint decoding and no MT cooperation equals the rate of the two-tap input-erasure Gaussian channel with i.i.d. inputs and $h_0=1,\ h_1=\alpha$.
\end{cor}
Since the shoft-handoff model is equivalent in terms of its per-cell sum-rate to the input-erasure channel, all the results and conclusions reported in Section \ref{sec: IEGC rate analysis}  also hold here. In particular the high-SNR characterization reduces to
\begin{equation}
\begin{aligned}
\label{ }
\Sinf&= (1-q)\ \ \mathrm{and}\\ 
\Linf &= -q^2\sum_{n=1}^\infty (1-q)^{n-1}\log\left(\frac{1-\alpha^{2n}}{1-\alpha^2}\right)\ .
\end{aligned}
\end{equation}
The benefits of MCP are evident in the high-SNR regime since its rate has two-fold degrees of freedom relative to the ICFS rate, and $1/q$-fold degrees of freedom relative to the SCP rate. 

It is noted that the MCP rate is an increasing function of the inter-cell interference factor $\alpha$. This is since: (a) $r$ and $s$ are increasing with $\alpha$ ; (b) $r>s$ ; and (c) by recalling that
\begin{equation}
\frac{r^{n+1}-s^{n+1}}{r-s}=\sum_{m=0}^{n}r^{n-m}s^{m}\ ,
\end{equation}
which is a summation of multiplications of increasing functions of $\alpha$ and hence is an increasing function of $\alpha$. 

It is also worth mentioning that a similar analysis in \cite{Levy-Zeitouni-Shamai-2Tap_UP08} facilitates also to address the issue of successful transmission per a given user, and when the average rate per user is examined, the results are in agreement with Prop. \ref{eq: two tap IEGC rate} (or Corr.  \ref{corr: MCP corr}).

\subsection{Numerical Results}


The achievable rates of the MCP, SCP and ICFS schemes with $\alpha^2=0.5$ and or $P=14$ [dB], are plotted as functions of the erasure rate $q$, in Figure \ref{fig: Comp 2}. The figure demonstrate the superiority of MCP over both SCP and ICFS schemes. It is also observed that ICFS is beneficial over SCP for erasure rates smaller than a certain threshold.  


\begin{figure}
\begin{center}
\includegraphics[scale = \myscale]{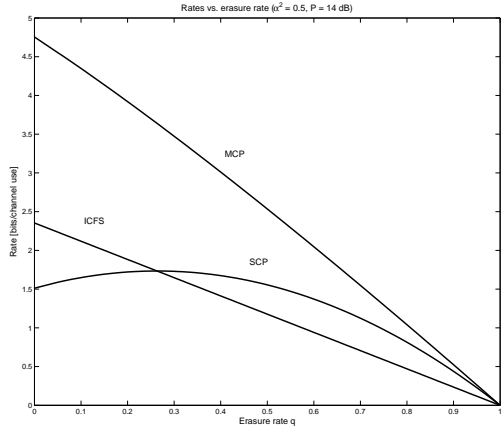}
\vspace{0cm}
\caption{Rates vs. erasure rate, for $\alpha^2=0.5$ and $P=14$ [dB].} \label{fig: Comp 2}
\end{center}
\end{figure}

\subsection{Discussion}

A different scenario which gives rise to a similar, but \textit{non-ergodic}, erasure model in a cellular scenario is one in which the MTs may be active or not in a given block with a given probability accounting for dynamic user activity. According to this model, each user is randomly selected to be active with probability $(1-q)$ throughout the entire transmission block. This may be implemented by means of control, or simply when the MT has no information to send. In these cases, the erasures are still i.i.d. along the cells' indices but are non-ergodic along the time index for each user. With optimal MCP, the \emph{throughput} \emph{per active user} is an r.v., given in the large system limit by
\begin{multline}
\label{eq: MCP rate user activity}
C_{\mathrm{mcp-ua}}=\lim_{N\rightarrow\infty}\\ \frac{1}{\tr{(\Mat{E}_N)}}\log \det{\left(\Mat{I}_{N+1}+
P\Mat{H}_{N+1}\Mat{E}_N\Mat{E}^\dagger_N\Mat{H}^\dagger_{N+1}\right)}\ ,
\end{multline}
where $\diag(\Mat{E}_N)=(e_1,\ldots,e_n)$ is a realization of the i.i.d. user-activity pattern.  Since $\frac{\tr{(\Mat{E}_N)}}{N}\xrightarrow[a.s.]{N\rightarrow\infty}(1-q)$ due to the strong law of large numbers, we can replace $\tr{(\Mat{E}_N)}$ with $(1-q)N$ in \eqref{eq: MCP rate user activity}. Following similar argumentations as those made in \cite{Levy-Zeitouni-Shamai_CLT-UP08}, it can be shown that $C_{\mathrm{mcp-ua}}$ converges a.s. to $(1-q)^{-1}C_{\mathrm{mcp}}$. Turning to the SCP and ICFS schemes it is easily verified that the throughputs per active user converge a.s. to the corresponding rates of the shadowing setup times $(1-q)^{-1}$.

\begin{figure}
\begin{center}
\includegraphics[scale = \myscale]{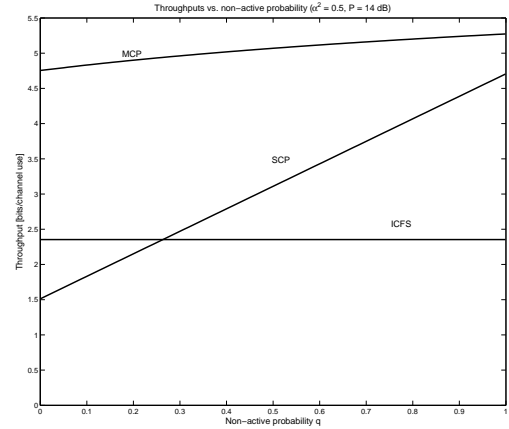}
\vspace{0cm}
\caption{Throughputs per  active user vs. non-active probability $q$, for $\alpha^2=0.5$ and $P=14$ [dB].} \label{fig: Comp 3}
\end{center}
\end{figure}

In Figure \ref{fig: Comp 3} the throughputs per active user of the MCP, SCP, and ICFS schemes are plotted as functions of the non-active probability $q$, for $\alpha^{2}=0.5$ and $P=14$ [dB]. The superiority of MCP over the other schemes is apparent. Also we can observe that both MCP and ICFS are robust and maintain their throughput constant (for ICFS) or almost constant (for MCP) regardless of $q$, while the throughput of SCP drops sharply with decreasing $q$. This is easily explained since neither MCP nor ICFS suffers from inter-cell interference, the first due to the fact that it leverages all received signals as useful, while the second since it avoids interference by frequency planning. On the other hand, SCP treats the other cell's signals as noise and hence, with more active users, more cells have to decrease their rate to overcome the emerging interference.

\section{Concluding Remarks}

In this work we have considered the input-erasure Gaussian channel. We have shown that with i.i.d. inputs, the input- and output-erasure \cite{Tulino-Verdu-Caire-Shamai_ISIT07} channels are equivalent. Focusing on the two-tap case, we have presented an alternative analysis to the one
reported in \cite{Tulino-Verdu-Caire-Shamai_ISIT07}. In particular, in contrast to \cite{Tulino-Verdu-Caire-Shamai_ISIT07} where the rate expression involves a fixed-point equation, the new expression is given in the form of an infinite sum, which enables simple calculation of arbitrarily accurate lower bounds. Moreover, the analysis is based on simple properties of tri-diagonal matrices and also holds for erasures with memory. An alternative simpler high-SNR rate power offset expression in the form of an infinite sum has also been presented, which contrasts to the one reported in \cite{Tulino-Verdu-Caire-Shamai_ISIT07} which involves a fixed-point equation. Numerical results have demonstrated the deleterious effects of increasing erasure rate and reveal the nature of the worst-case filter (coefficients with equal amplitudes) and that of the best-case filter (either coefficients is set to zero). The results have then been used to demonstrate the benefits of MCP over SCP schemes for a simple cellular uplink with shadowing fading or with user-activity model. Other settings that include independent shadowing to each BS, and more complex user-activity models, are currently being studied.

\appendix


We need to calculate
\begin{multline}
\label{eq: two-tap IEGC rate}
R=\lim_{N\rightarrow\infty}\\ \E\left[\frac{1}{N}\log \det{\left(\Mat{I}_{N+1}+P\Mat{H}_{N+1}\Mat{E}_N\Mat{E}^\dagger_N\Mat{H}^\dagger_{N+1}\right)}\right]\  .
\end{multline}
Using \eqref{eq: IEGC-OEGC equivalence}
and defining the $N\times N$ matrix
\begin{equation}
\label{eq: Cov matrix def}
\Mat{G}_N \triangleq  \Mat{E}^\dagger_N\Mat{H}^\dagger_{N+1}\Mat{H}_{N+1}\Mat{E}_N\ ,
\end{equation}
we can rewrite \eqref{eq: two-tap IEGC rate} as
\begin{equation}
\label{eq: two-tap IEGC rate G}
C=\lim_{N\rightarrow\infty}\E\left(\frac{1}{N}\log \det{\left(\Mat{I}_N+P\Mat{G}_N\right)}\right)\ .
\end{equation}

Straightforward calculations reveal that $\Mat{G}_N$ is given explicitly by
\begin{equation}
\label{ }
[\Mat{G}_N]_{n,m}=\left\{\begin{array}{cc}e_n(\abs{h_0}^2+\abs{h_1}^2) & n=m \\e_ne_{n+1} h_0 h_1^\dagger & n=m-1 \\e_{n-1}e_n h_0^\dagger h_1 & n=m+1 \\0 & \mathrm{otherwise}\end{array}\right.\ ,
\end{equation}
e.g., for the special case of $N=3$ we have
\begin{multline}
\label{ }
\Mat{G}_3 =\\ \scriptsize\left(\begin{array}{ccc}
     e_1(\abs{h_0}^2+\abs{h_1}^2) & e_1e_2 h_0 h_1^\dagger & 0 \\
     e_1e_2 h_0^\dagger h_1 & e_2(\abs{h_0}^2+\abs{h_1}^2) & e_2e_3 h_0 h_1^\dagger \\
     0 & e_2e_3 h_0^\dagger h_1 & e_3(\abs{h_0}^2+\abs{h_1}^2)
\end{array}\right).
\end{multline}
It is observed that $\Mat{G}_N$ is a tri-diagonal matrix (due to the two-tap filter), and that  an erasure of the $n$th input symbol, i.e., $e_n=0$, sets the $n$th row and the $n$th column of $G_N$ to zero. Hence, each erasure splits $\Mat{G}_N$ into two block-diagonal sub-matrices. The latter is easily explained by the following example. Assume $N=10$ and an erasure realization $\vct{e}'=(1,1, 0, 0,1,1,1,0,1,1)$. In this case we have
\begin{multline}
\label{eq: big mat}
\Mat{I}_{10}+P\Mat{G}_{10} = \\
\left(\begin{array}{cccccccccc}
a & c & 0 & 0 & 0 & 0 & 0 & 0 & 0 & 0\\
c^\dagger & a & 0 & 0 & 0 & 0 & 0 & 0 & 0 & 0\\
0 & 0 & 1 & 0 & 0 & 0 & 0 & 0 & 0 & 0\\
0 & 0 & 0 & 1 & 0 & 0 & 0 & 0 & 0 & 0\\
0 & 0 & 0 & 0 & a & c & 0 & 0 & 0 & 0\\
0 & 0 & 0 & 0 & c^\dagger & a & c & 0 & 0 & 0\\
0 & 0 & 0 & 0 & 0 & c^\dagger & a & 0 & 0 & 0\\
0 & 0 & 0 & 0 & 0 & 0 & 0 & 1 & 0 & 0\\
0 & 0 & 0 & 0 & 0 & 0 & 0 & 0 & a & c\\
0 & 0 & 0 & 0 & 0 & 0 & 0 & 0 & c^\dagger & a\\
\end{array}\right)\ ,
\end{multline}
where $a$ is defined in \eqref{eq: a b definitions}, and $c\triangleq Ph_0^\dagger h_1$.
Examining \eqref{eq: big mat} it is evident that
\begin{multline}
\label{ }
\log \det{(\Mat{I}_{10}+P\Mat{G}_{10})}=\\ 2 \log \det{(\Mat{I}_{2}+P\Mat{\bar{G}}_{2})} + \log \det{(\Mat{I}_{3}+P\Mat{\bar{G}}_{3})}\ ,
\end{multline}
where $\Mat{\bar{G}}_{n}\triangleq\Mat{H}^\dagger_{n+1}\Mat{H}_{n+1}$ is a tri-diagonal $n\times n$ Toeplitz matrix (equal to $\Mat{G}_n$ with no erasures). Turning to the general case
\begin{multline}\label{eq: two-tap IEGC term with sums}
\log\det{(\Mat{I}_{N}+P\Mat{G}_{N})} =\\  \sum_{n=1}^N  n(\vct{e'}) \log\det{(\Mat{I}_{n}+P\Mat{\bar{G}}_{n})},
\end{multline}
where $n(\vct{e}')$ is the number of $n$-length consecutive non-erasure sub-sequences in $\vct{e}'$. Using \eqref{eq: two-tap IEGC term with sums}, we can express the rate for any block size $N$ as 
\begin{equation}
\label{eq: rate N raw}
\begin{aligned}
R_N&=\E\left[\sum_{n=1}^N  \frac{n(\vct{e})}{N} \log\det{(\Mat{I}_{n}+P\Mat{\bar{G}}_{n})}\right]\\
&=\sum_{n=1}^N  \E\left[\frac{n(\vct{e})}{N}\right] \log\det{(\Mat{I}_{n}+P\Mat{\bar{G}}_{n})}\ .
\end{aligned}
\end{equation}
Defining $\mathrm{1}_{\{e:\ n,m\}}$ as an indicator function of the event that an isolated sequence of $n$ non-erasures starts at index $m$ in $\vct{e}$ we get that
\begin{equation}
\label{eq: consec ones prob}
\E\left[\frac{n(\vct{e})}{N}\right]=\E\left[\frac{1}{N}\sum_{m=1}^N \mathbf{1}_{\{e:\ n,m\}}\right]=
\left(1-\frac{n}{N}\right)q^2(1-q)^n\ ,
\end{equation}
where the last equality is due to the fact that the erasures are i.i.d, hence\footnote{For simplicity maters we neglect the fact that for $m=1$, $\E\left[\mathbf{1}_{\{e:\ n,1\}}\right]=q(1-q)^n$ and claime it has no effect in the large system limit.}
\begin{equation}
\label{ }
\E\left[\mathbf{1}_{\{e:\ n,m\}}\right] =\left\{\begin{array}{cc}q^2(1-q)^n & m=2,\ldots,N-n \\
0 & \mathrm{otherwise}\end{array}\right.\ .
\end{equation}

Next, we  calculate $\det{(\Mat{I}_{n}+P\Mat{\bar{G}}_{n})}\ ;\ n=1,\ldots,N$, by recalling the fact that the determinant of any tri-diagonal ${n \times n}$ matrix ${\Mat{V}}$ obeys the following \emph{recursive} relation
\begin{multline}
\label{eq: recursive relations}
 \det{\Mat{V}_n}=  [\Mat{V}_n]_{n,n}\det{\Mat{V}_{n-1}}\\ - [\Mat{V}_n]_{n,n-1} [\Mat{V}_n]_{n-1,n}\det{\Mat{V}_{n-2}}\ \  ;\ \ n\ge 3\ ,
\end{multline}
where $\Mat{V}_{n-1}$, $\Mat{V}_{n-2}$ are the first two main \emph{principal sub-matrices} of $\Mat{V}_{n}$ (composed of the first $n-1$ and $n-2$ rows and columns of $\Mat{V}_n$, respectively). Applying the recursive formula to  $\Mat{D}_n \triangleq \Mat{I}_{n}+P\Mat{\bar{G}}_{n}$ we get the following \emph{difference} equation:
\begin{equation}
\label{ }
\det{(\Mat{D}_n)}=a\det{(\Mat{D}_{n-1})}-b^2\det{(\Mat{D}_{n-2})}\ ;\ n\ge2\ ,
\end{equation}
with initial conditions $\det{(\Mat{D}_0)} = 1$ and $\det{(\Mat{D}_1)} = a$, where $a$ and $b$ are defined in \eqref{eq: a b definitions}. A solution to a difference equation of this form is given by
\begin{equation}
\label{ }
 \det{(\Mat{D}_n)} = \varphi r^n - \phi s^n\ ,
\end{equation}
where $r$ and $s$ are defined in \eqref{eq: r s definitions}. The initial condition are used to calculate $\varphi$ and $\phi$:
\begin{equation}
\label{eq: diff consts}
\varphi = \frac{r}{r-s}\quad\mathrm{and}\quad \phi=\frac{s}{r-s}\  ;
\end{equation}
hence,
\begin{equation}
\label{eq: diff solution}
\det{(\Mat{D}_n)}=\frac{r^{n+1}-s^{n+1}}{r-s}\ ;\ n=1,\ldots,N\ .
\end{equation}
Substituting \eqref{eq: consec ones prob} and \eqref{eq: diff solution} into \eqref{eq: rate N raw} we get that
\begin{equation}
\label{ }
R_N=\sum_{n=1}^N q^2(1-q)^2\left(1-\frac{n}{N}\right)\log \left(\frac{r^{n+1}-s^{n+1}}{r-s}\right)\ .
\end{equation}

It remains to show that $R_N\xrightarrow{N\rightarrow\infty}R$. To do that we notice that
\begin{multline}
\label{eq: Finite rate with R}
R_N = R - \frac{1}{N}\sum_{n=1}^N q^2(1-q)^n n \log\det{(\Mat{D}_n)}\ - \\ \sum_{n=N+1}^\infty q^2(1-q)^n \log\det{(\Mat{D}_n)}\ .
\end{multline}
Hence, $R_N$ can be bounded by
\begin{equation}
\label{ }
R - \frac{\beta}{N}\sum_{n=1}^\infty q^2(1-q)^n\ n^2 - \beta \sum_{n=N}^\infty q^2(1-q)^n n \le R_N \le R\ .
\end{equation}
This is since the summands in \eqref{eq: Finite rate with R} are positive, and by noticing that  $\log \det{(\Mat{D}_n)}\le n\beta$ (for some positive $\beta>0$). After some algebra we get that 
\begin{equation}
\label{ }
R - \frac{2\beta(1-q)}{Nq} - \beta(1-q)^N\left((N+1)q+1)\right) \le R_N \le R\ .
\end{equation}
Finally, taking $N\rightarrow\infty$ completes the proof. 


\section*{Acknowledgment}

This work was supported by a Marie Curie Outgoing International
Fellowship and the NEWCOM++ network of excellence both within the 6th and 7th
European Community Framework Programmes, the National Science Foundation under
Grants CNS-06-26611 and CNS-06-25637, and the REMON consortium for wireless
communication.


\bibliographystyle{IEEEtran}

\end{document}